\documentclass[envcountsame,envcountsect]{svmult}

\usepackage{mathptmx}       
\usepackage{helvet}         
\usepackage{courier}        
\usepackage{type1cm}        
\usepackage[square,numbers,sort&compress]{natbib}
\usepackage{epsfig}
\usepackage{graphicx}        
\usepackage{multicol}        
\usepackage[bottom]{footmisc}
\usepackage{url}
\usepackage{latexsym}
\usepackage{amssymb,amsmath}
\usepackage{mathrsfs}
\usepackage{amsfonts}
\usepackage[latin1]{inputenc}
\usepackage{stmaryrd}
\usepackage{eurosym}
\usepackage{tikz}
\usepackage{textcomp}
\usepackage{blindtext}
\usepackage{tabularx}

\newcommand{\bbn}{\mathbb{N}}

\newcommand{\bbr}{\mathbb{R}}
\newcommand{\bbc}{\mathbb{C}}

\begin{document}

\title*{Derivative pricing under the possibility of long memory in the supOU stochastic volatility model}
\titlerunning{Derivative pricing in the supOU SV model} 
\author{Robert Stelzer and  Jovana Zavi\v{s}in}
\institute{Robert Stelzer\at Institute of Mathematical Finance, Ulm University, Helmholtzstr. 18, D-89081 Ulm, Germany, \email{robert.stelzer@uni-ulm.de}
\and Jovana Zavi\v{s}in \at \email{jovana.zavisin@gmail.com}}
%
%
\maketitle

\abstract{We consider the supOU stochastic volatility model which is able to exhibit long-range dependence. For this model we give conditions for the discounted stock price to be a martingale, calculate the characteristic function, give a strip where it is analytic and discuss the use of Fourier pricing techniques. \\
Finally, we present a concrete specification with polynomially decaying autocorrelations and calibrate it to observed market prices of plain vanilla options.}

\textbf{AMS Subject Classification 2010: }\\
\begin{tabular}{ll}
Primary: & \quad 91G20 , 60G51 \\
Secondary: & \quad 91B25
\end{tabular}
\keywords{calibration, Fourier pricing, L\'evy basis, long memory, superposition of Ornstein-Uhlenbeck type processes, stochastic volatility}
\section{Introduction}
The Ornstein-Uhlenbeck (OU) type stochastic volatility (SV) model introduced in \cite{Barndorffetal2001c} is one of the most popular stochastic volatility models for prices of financial assets driven by a L\'evy process (see e.g. \cite{Schoutens2003,Contetal2004}). It covers many of the stylized facts typically encountered in financial data (cf. \cite{Guillaumeetal1997,Cont2001}). Over the years many variants have been introduced, for instance a variant with two sided jumps in \cite{WILM:WILM10217} or a multivariate extension in \cite{pigorsch:stelzer:2009}. 

In this paper we consider a variant of the model which additionally can cover the stylized fact of long-range dependence (or slower than exponentially decaying autocorrelations), the supOU stochastic volatility model. In this model we specify the volatility as a superposition of Ornstein-Uhlenbeck (thus ``supOU'') processes, which have been introduced in \cite{barndorff:2001}. Various features of this volatility model (in a multidimensional setting) have been considered in   \cite{BarndorffetStelzer2009sv,barndorff:stelzer:2011a,moser:stelzer:2011,StelzerTosstorff2011}.

The focus of the present paper is on derivative pricing in and calibration of the univariate supOU SV model similar to the papers \cite{MuhleKarbePfaffelStelzer2009,NicolatoVenardos2003} in the (multivariate) OU type SV model. To this end we first briefly review the model in Section \ref{sec:1}. In Section \ref{sec:2} we give conditions on the parameters such that the discounted stock price process is a martingale which implies that under these conditions the model can be used to describe the risk neutral dynamics of a financial asset. Thereafter, we start Section~\ref{sec:3} with a review of Fourier pricing. Then, we give the characteristic function of the log asset price in the supOU SV model and show conditions for the moment generating function to be sufficiently regular so that Fourier pricing is applicable. 
Finally, we present a concrete specification, the $\Gamma$-supOU SV model, in Section \ref{sec:4} and discuss its calibration to market data which we illustrate with a small example using options on the DAX. Finally, we discuss a subtle issue regarding how to employ the calibrated model to calculate prices of European options with a general maturity.

\section{A review of the supOU stochastic volatility model}\label{sec:1}
We briefly review the definition and the most important known facts of the supOU stochastic volatility model introduced in \cite{BarndorffetStelzer2009sv}. More background on supOU processes can be found in \cite{StelzerTosstorff2011,FasenetCklu2007,barndorff:2001,barndorff:stelzer:2011a}.

In the following $\mathbb{R}_{-}$ denotes the set of negative real numbers and $\mathcal{B}_b(\mathbb{R}_{-}\times \mathbb{R})$ denotes the bounded Borel sets of $\mathbb{R}_{-}\times \mathbb{R}$.

\begin{definition}
A family $\Lambda = \{\Lambda(B): B\in\mathcal{B}_b(\mathbb{R}_{-}\times \mathbb{R})\}$ of real-valued random variables is called a real-valued L\'{e}vy basis (infinitely divisible independently scattered random measure) on $\mathbb{R}_{-}\times \mathbb{R}$ if:

\begin{itemize}
\item the distribution of $\Lambda(B)$ is infinitely divisible for all $B\in\mathcal{B}_b(\mathbb{R}_{-}\times \mathbb{R})$,
\item for any $n\in\mathbb{N}$ and pairwise disjoint sets $B_1,...,B_n \in \mathcal{B}_b(\mathbb{R}_{-}\times \mathbb{R})$ the random variables $\Lambda(B_1),...,\Lambda(B_n)$ are independent,
\item for any sequence of pairwise disjoint sets $B_n\in\mathcal{B}_b(\mathbb{R}_{-}\times \mathbb{R})$ with $n\in\mathbb{N}$ satisfying $\cup_{n\in\mathbb{N}} B_n \in \mathcal{B}_b(\mathbb{R}_{-}\times \mathbb{R})$ the series $\sum_{n=1}^{\infty} \Lambda(B_n)$ converges a.s. and $\Lambda(\cup_{n\in\mathbb{N}}B_n) = \sum_{n=1}^{\infty}\Lambda(B_n)$.
\end{itemize}
\end{definition}

We consider only L\'{e}vy bases with characteristic functions of the form
\begin{align*}
 \mathbb{E}(\exp (iu\Lambda(B))) = \exp(\varphi(u)\Pi(B))
\end{align*}
for all $u\in\mathbb{R}$ and all $B\in \mathcal{B}_b(\mathbb{R}_{-}\times \mathbb{R})$, where $\Pi = \pi \times \lambda$ is the product of a probability measure $\pi$ on $\mathbb{R}_{-}$ and the Lebesgue measure $\lambda$ on $\mathbb{R}$ and
\begin{align*}
\varphi(u) = iu\gamma_0 +\int_{\mathbb{R_+}}\bigg{(} e^{iux}-1 \bigg{)} \nu(dx)
\end{align*}
is the cumulant transform of an infinitely divisible distribution on $\bbr_+$ with L\'{e}vy-Khintchine triplet $(\gamma_0,0,\nu)$, which is also the characteristic triplet of the underlying L\'evy process $L_t = \Lambda(\mathbb{R}_{-}\times (0,t]) \text{ and } L_{-t} = \Lambda(\mathbb{R}_{-}\times (-t,0))$ for $t\in\bbr_+$ (see e.g. \cite{Sato1999} for the relevant background on infinitely divisible distributions and L\'evy processes). We call the triplet $(\gamma_0,\nu,\pi)$ the \emph{generating triplet}. Note that this means that $\gamma_0\geq 0$, $\nu(\bbr \backslash \bbr_+) = 0$ and $\int_{|x| \leq 1} |x| \nu(dx) < \infty$.

If $L$ is a pure jump L\'evy process with triplet $(0,0,\nu)$ and jump measure $N(ds,dx)$, then turning the Poisson point process of jumps in $\bbr\times \bbr_+\backslash\{0\}$ to one in $\bbr\times \bbr_+\backslash\{0\}\times \bbr_-$ by marking all jumps with independent marks distributed according to $\pi$ produces the jump measure of a L\'evy basis with triplet $(\gamma_0,\nu,\pi)$.

In the supOU process defined now this can be understood as marking every jump of a L\'evy process with an individual exponential decay rate. We restrict our attention to positive supOU processes as this is natural when using them to model a variance changing over time.

\begin{theorem}
Let $\Lambda$ be an $\bbr_+$-valued L\'{e}vy basis on $\bbr_- \times \bbr$ with generating triplet $(\gamma_0, \nu, \pi)$. Assume
 \begin{equation*}
  \int_{|x|>1} \ln(|x|) \nu(dx) < \infty,  \quad \text{and} \quad -\int_{\bbr_-} \frac{1}{A} \pi(dA) < \infty.
 \end{equation*}
 Then the process $\Sigma = (\Sigma_t)_{t \in \bbr}$ given by
 \begin{flalign*}
  \Sigma_t & = \int_{\bbr_-} \int_{-\infty}^t e^{A(t-s)} \Lambda(dA, ds) 
 \end{flalign*}
is well-defined as a Lebesgue integral for all $t \in \bbr$  and it is stationary.

Moreover, $\Sigma_t \geq 0$ for all $t \in \bbr$ and the distribution of $\Sigma_t$ is infinitely divisible with characteristic function given by
$
 \mathbb{E} \left( e^{ iu\Sigma_t} \right) = e^{ iu\gamma_{\Sigma, 0} + \int_{\bbr_+} \left( e^{iux} -1 \right) \nu_\Sigma (dx)}
$
for all $u \in \bbr$ where
$
 \gamma_{\Sigma, 0} = \int_{\bbr_-} \int_0^\infty e^{As} \gamma_0 ds \pi(dA),
\,\,\,
 \nu_\Sigma (B) = \int_{\bbr_-} \int_0^\infty \int_{\bbr_+} \mathbf{1}_B \left( e^{As}x \right) \nu(dx) ds \pi(dA)
$
for all $B \in \mathscr{B} (\bbr)$.
\end{theorem}
As shown in \cite[Th. 3.12]{barndorff:stelzer:2011a} the supOU process is adapted to the filtration generated by $\Lambda$ and has locally bounded paths. Provided $\pi$ has a finite first moment, one can take a supOU process to have c\`adl\`ag paths.

\begin{definition}\label{supOUsv}
 Let $W$ be a standard Brownian motion, $a = (a_t)_{t \in \bbr_+}$ a predictable real-valued process, $\Lambda$ an $\bbr_+$-valued L\'{e}vy basis on $\bbr_- \times \bbr$ independent of $W$ with generating triplet $(\gamma_0, \nu, \pi)$ and let $L$ be its underlying L\'{e}vy process. Let $\Sigma$ be a non-negative c\`adl\`ag supOU process and $\rho \in \bbr$. Assume that $X = (X_t)_{t \in \bbr_+}$ is given by 
 \begin{equation*}
  X_t = X_0 + \int_0^t a_s ds + \int_0^t \Sigma_s^{\frac{1}{2}} dW_s + \rho (L_t-\gamma_0 t),
 \end{equation*}
 where $X_0$ is independent of $\Lambda$. Then we say that $X$ follows a univariate supOU stochastic volatility model and refer to it by $SVsupOU(a, \rho, \gamma_0, \nu, \pi)$.
\end{definition}

When we speak about properties related to a filtration above or in the following we refer to the filtration generated by $W$ and $\Lambda$.

Above $X$ is supposed to be the log price of some financial asset and $\rho$ is the typically negative correlation between jumps in the volatility and log asset prices modelling the leverage effect. To ensure that the absolutely continuous drift is completely given by $a_t$ we subtract the drift $\gamma_0$ from the L\'evy process noting that this can be done without loss of generality.

In \cite{BarndorffetStelzer2009sv} it has been shown that  the model is able to exhibit long-range dependence in the squared log-returns. The typical example leading to a polynomial decay of the autocovariance function of the squared returns and to long-range dependence for certain choices of the parameter  is to take $\pi$ as  a Gamma distribution mirrored at the origin. \cite{FasenetCklu2007,StelzerTosstorff2011} discuss in general which properties of $\pi$ result in long-range dependence.
\section{Martingale conditions}\label{sec:2}
Now we assume given a market with a deterministic numeraire (or bond) with price process $e^{rt}$ for some $r\geq 0$ and a risky asset with price process $S_t$. 

We want to model the market by a supOU stochastic volatility model under the risk neutral dynamics. Thus we need to understand when $\hat S_t=e^{-rt} e^{X_t}$ is a martingale for the filtration $\mathbb{G}=(\mathcal{G}_t)_{t\in\bbr_+}$ generated by the Wiener process and the L\'evy basis, i.e. $\mathcal{G}_t=\sigma\left( \left\{\Lambda(A),\, W_s: s\in [0,t] \mbox{ and  } A\in \mathscr{B}_b(\bbr_-\times (-\infty,t])\right\}\right)$ for $t\in\bbr_+$. Implicitly we understand the filtration is modified such that the usual hypotheses (see e.g. \cite{Protter2004}) are satisfied.
\begin{theorem}[Martingale condition] Consider a market as described above.
 Suppose that
 \begin{equation}
  \int_{ x>1 } \left( e^{\rho x} - 1 \right) \nu(dx) < \infty.
 \end{equation}
 If the process $a = (a_t)_{t \in \bbr_+}$ satisfies
 \begin{equation}\label{martingale_cond}
  a_t = r - \frac{1}{2} \Sigma_t - \int_{\bbr_+} \left( e^{\rho x} - 1 \right) \nu(dx),
 \end{equation}
then the discounted price process $\hat{S}$ is a martingale.
\end{theorem}
\begin{proof}
The arguments are straightforward adaptations of the ones in \cite[Prop. 2.10]{MuhleKarbePfaffelStelzer2009} or \cite[Sec. 3]{NicolatoVenardos2003}.
\end{proof}

\section{Fourier pricing in the supOU stochastic volatility model}\label{sec:3}
Our aim now is to use the Fourier pricing approach in the supOU stochastic volatility model for calculating prices of European derivatives.
\subsection{A review on Fourier pricing}
We start with a brief review on the well-known Fourier pricing techniques introduced in \cite{Raible2000,CarrMadan1999}.

Let the price process of a financial asset be modeled as an exponential semimartingale $S=(S_t)_{0 \leq t \leq T}$ i.e.
$
 S_t = S_0 e^{X_t}, \quad 0 \leq t \leq T
$
where $X = (X_t)_{0 \leq t \leq T}$ is a semimartingale.

 Let $r$ be the risk-free interest rate and let us assume that we are directly working under an equivalent martingale measure, i.e. the discounted price process $\hat{S} = (\hat{S}_t)_{0 \leq t \leq T}$ given by $\hat{S}_t = S_0 e^{X_t - rt}$ is a martingale.

We call the process $X$ the underlying process and without loss of generality we can assume that $X_0 = 0$. We denote by $s$ minus the logarithm of the initial value of $S$, i.e. $s = -\log(S_0)$.

Let $\hat{f}$ denote the Fourier transform of the function $f$, i.e. $
 \hat{f}(u) = \int_{\bbr} e^{iux} f(x) dx.
$

Let now $f : \bbr_+ \rightarrow \bbr$ be a measurable function that we refer to as the payoff function. Then, the arbitrage-free price of the derivative with payoff $f(X_T - s)$ and maturity $T$  at time zero is the conditional expected discounted  payoff under the chosen equivalent martingale measure, i.e.
$
 V_f(X_T;s) = e^{-rT} \mathbb{E} \left( f(X_T - s) |\mathcal{G}_0\right) .
$

The following theorem gives the valuation formula for the price of the derivative paying $f(X_T - s)$ at time $T$.
\begin{theorem}  \cite[Th. 2.2, Rem. 2.3]{EGP2010} \label{fourier1}
 Let $f: \bbr_+ \rightarrow \bbr$ be a payoff function and let $g_R(x) = e^{-Rx} f(x)$ for some $R \in \bbr$ denote the dampened payoff function. Define $\Phi_{X_T|\mathscr{G}_0} (u)  := \mathbb{E} \left( e^{uX_T} |\mathscr{G}_0 \right), \quad u \in \bbc$. If
  \[ (i) \quad g_R \in L^1(\bbr)\cap L^\infty(\bbr),\quad (ii) \quad \Phi_{X_T|\mathscr{G}_0} (R)<\infty,\quad (iii)\quad \Phi_{X_T|\mathscr{G}_0}(R+i\cdot)\in L^1(\bbr),\]
 then $
  V_f(X_T;s) = \frac{e^{-rT - Rs}}{2 \pi} \int_\bbr e^{-ius} \Phi_{X_T|\mathscr{G}_0} (R + iu) \hat{f}(iR - u) du.
$
\end{theorem}
It is well known that for a European Call  option with maturity $T$ and strike $K>0$ condition $(i)$ is satisfied for $R>1$ and that for the payoff function 
  $
    f(x) = \max(e^x - K, 0) =:  (e^x - K)^+
$
   the Fourier transform is 
  $
    \hat{f}(u) = \frac{K^{1 + iu}}{iu (1 + iu)}
 $ for $u \in \bbc$ with $\operatorname{Im}(u) \in (1, \infty)$.

In the following we calculate the characteristic/moment generating function for the supOU SV model and show conditions when the above Fourier pricing techniques are applicable.
\subsection{The characteristic function}
Consider the general supOU SV model with drift  of the form $a_t=\mu+\gamma_0+\beta \Sigma_t$. Note that then the discounted stock price is a martingale if and only if $\beta=-1/2$ and $\mu+\gamma_0=r-\int_{\bbr_+} \left( e^{\rho x} - 1 \right) \nu(dx)$. 

Standard calculations as in \cite[Th. 2.5]{MuhleKarbePfaffelStelzer2009} or \cite{NicolatoVenardos2003} give the following result which is the univariate special case of a formula reported in \cite[Sec. 5.2]{barndorff:stelzer:2011a}.
\begin{theorem}\label{theorem_ch_fn}
 Let $X_0 \in \bbr$ and let the log-price process $X$ follow a supOU SV model of the above form. Then, for every $t \in \bbr_+$ and for all $u \in \bbr$ the characteristic function of $X_t$ given $\mathscr{G}_0$ is given by
 \begin{align}\label{chfn}
  &\Phi_{X_t|\mathscr{G}_0} (iu)=\mathbb{E} \left( e^{iuX_t} | \mathscr{G}_0 \right)=\\&\quad = \exp \Bigg\{ i \Bigg( u(X_0 + \mu t) 
  +  \left( u \beta + \frac{i}{2} u^2 \right) \int_{\bbr_-} \int_{-\infty}^0 \frac{1}{A} \left( e^{A(t-s)} - e^{-As} \right) \Lambda(dA, ds) \Bigg) \nonumber\\
 &\quad\quad +  \int_{\bbr_-}  \int_0^t \varphi \left( \frac{e^{A(t-s)}}{A} \left( u \beta + \frac{i}{2} u^2 \right) - \left( \frac{1}{A} \left( u \beta + \frac{i}{2} u^2 \right) - \rho u \right) \right) ds \pi(dA) \Bigg\}.\nonumber
 \end{align}
\end{theorem}

Note that in contrast to the case of the OU type stochastic volatility model, where $(X, \Sigma)$ is a strong Markov process, in the supOU stochastic volatility model $\Sigma$ is  not Markovian. Thus, conditioning on $X_0$ and $\Sigma_0$ is not equivalent to conditioning upon $\mathscr{G}_0$. Therefore $\Phi_{X_t|\mathscr{G}_0} (iu)$ is not simply a function of $X_0,\Sigma_0$. Instead, the whole past of the L\'evy basis enters via the $\mathcal{G}_0$-measurable 
\begin{equation*}
  z_t := \int_{\bbr_-} \int_{-\infty}^0 \frac{1}{A} \left( e^{A(t-s)} - e^{-As} \right) \Lambda(dA, ds),
\end{equation*}
which has a similar role as the initial volatility $\Sigma_0$ in the OU type stochastic volatility model. 
Like $\Sigma_0$ in the OU type models, $z_t$ can be treated as an additional parameter to be determined when calibrating the model to market option prices. We can immediately see that thus the number of parameters to be estimated increases with each additional maturity. As it will become clear later, the following observation is important.
\begin{lemma}\label{lem:z}
  $z_{t_1} \leq z_{t_2}$, for all $t_1,t_2 \in \bbr_+$ such that $t_1 \leq t_2$.
\end{lemma} 
\begin{proof}For $t \in \bbr_+$ and $s \leq t$ we have 
  $
    \frac{1}{A} \left( e^{A(t-s)} - e^{-As} \right) = \frac{e^{-As}}{A} \left( e^{At} - 1 \right)
$
  and for $t_1 \leq t_2$
  one sees $
    e^{A t_2} - 1 \leq e^{A t_1} - 1 \leq 0
  $
  since $A < 0$. This implies that for $s \leq t_1 \leq t_2$
 $
  \frac{e^{-As}}{A} \left( e^{A t_1} - 1 \right) \leq \frac{e^{-As}}{A} \left( e^{A t_2} - 1 \right)
$
  and thus $z_{t_1} \leq z_{t_2}$.
\end{proof}

\subsection{Regularity of the moment generating function}
In order to apply Fourier pricing we now show where the moment generating function $\Phi_{X_T|\mathcal{G}_0}$ is analytic.

Let $\theta_L(u) = \gamma_0 u + \int_{\bbr_+} \left( e^{ux} - 1 \right) \nu(dx)$ be the cumulant transform of the L\'evy basis (or rather its underlying subordinator). 
If
 $
  \int_{x \geq 1 } e^{rx} \nu(dx) < \infty \quad \text{for all } r \in \bbr \text{ such that } r < \epsilon
 $
  for some $\epsilon > 0$, then the function $\theta_L$ is analytic in the open set
  $
   S_L := \{ z \in \bbc: \quad \operatorname{Re}(z) < \epsilon \},
  $ as can be seen e.g. from the arguments at the start of the proof of \cite[Lemma 2.7]{MuhleKarbePfaffelStelzer2009}.
\begin{theorem}\label{strip}
  Let the measure $\nu$ satisfy
  \begin{equation}\label{condition}
  \int_{x \geq 1 } e^{rx} \nu(dx) < \infty \quad \text{for all } r \in \bbr \text{ such that } r < \epsilon
  \end{equation}
   for some $\epsilon > 0$. Then the function
   $
    \Theta(u) =  \int_{\bbr_-} \int_0^t \theta_L (u f_u(A,s) ) ds \pi(dA)
  $
    is analytic on the open strip
\begin{equation}
    S := \{ u \in \bbc, \quad |\operatorname{Re}(u)| < \delta \}
  \mbox{ with }
    \delta := -| \beta | - \frac{|\rho|}{t} + \sqrt{\Delta},\label{delta}
\end{equation}
   where
  $
    \Delta := \left( |\beta| + \frac{|\rho|}{t} \right)^2 +  \frac{2 \epsilon}{t}.
 $
 \end{theorem}
The rough idea of the proof is similar to \cite[Th. 2.8]{MuhleKarbePfaffelStelzer2009}, but the fact that we now integrate over the mean reversion parameter adds significant difficulty, as now bounds independent of the mean reversion parameter need to be obtained and a very general holomorphicity result for integrals has to be employed.
 \begin{proof} Define \begin{equation}\label{eq:f}
   f_u(A,s) = \mathbf{1}_{[0, t]} (s) \left( \frac{e^{A(t-s)}}{A} \left( \beta + \frac{u}{2} \right) - \left( \frac{1}{A} \left( \beta + \frac{u}{2} \right) - \rho \right) \right).
   \end{equation}
  We first determine $\delta > 0$ such that for all $u \in \bbr$ with $|u| < \delta$ it holds that $|u f_u(A, s)| < \epsilon$. We have
  \begin{flalign}\label{bound2}
   | u f_u(A, s)| & \leq \left| \frac{e^{A(t-s)} - 1}{A} \right| \left( |\beta| |u| + \frac{u^2}{2} \right) + |\rho| |u|
  \end{flalign}
  by the triangle inequality. In order to find the upper bound for the latter term, we first note that elementary analysis shows 
  \begin{equation}\label{newbound}
   \left| \frac{e^{A(t-s)} - 1}{A} \right| \leq t
  \end{equation}for all $A < 0$ and $s \in [0, t]$.
 Thus, we have to find $\delta > 0$ such that 
 $
   | u f_u(A, s) | \leq t \left( |\beta| |u| + \frac{u^2}{2} \right) + |\rho| |u| < \epsilon,
 $  for all $u \in \bbr$ with $|u| < \delta$,
  i.e. to find the solutions of the quadratic equation
  \begin{equation}\label{equation}
   \frac{t}{2} u^2 + \left( t |\beta| + |\rho| \right) |u| - \epsilon = 0.
  \end{equation}
 Since for $u = 0$ the sign of (\ref{equation}) is negative, i.e. (\ref{equation}) is equal to $-\epsilon$, we know that there exist one positive and one negative solution. The positive one is $\delta$ as given in (\ref{delta}).

Now let $u \in S$, i.e. $u = v + iw$ with $v, w \in \bbr$, $|v| < \delta$. Observe that
  $ 
   \operatorname{Re} ( u f_u(A,s))   = v f_v(A,s) - \frac{w^2}{2} \left( \frac{e^{A(t-s)}-1}{A}\right)
$
  and $\frac{e^{A(t-s)}-1}{A} \geq 0$ for all $s \in [0, t]$ and $A < 0$. Hence, $\operatorname{Re} ( u f_u(A,s)) \leq v f_v(A,s)$. This implies that
\[
   \int_{ x \geq 1  } e^{\operatorname{Re} ( u f_u(A,s)) x} \nu(dx) \leq \int_{x \geq 1  } e^{ v f_v(A,s) x} \nu(dx) < \infty
\]
  due to $ |v f_v(A,s)| < \epsilon $ for $|v| < \delta$ and condition (\ref{condition}). Hence for $u \in S$ the function
$
   \theta_L( u f_u(A, s)) =  \gamma_0 uf_u(A, s)  + \int_{\bbr_+} \left( e^{u f_u(A,s) x} -1 \right) \nu(dx)
$
  is well-defined. $u f_u(A,s)$ is a polynomial of $u$ and thus it is an analytic function in $\bbc$ for all $s \in [0,t]$ and $A < 0$. The function $\theta_L$ is analytic in the set
  $
  S_L = \{ z \in \bbc: \quad |\operatorname{Re}(z)| < \epsilon \}.
  $

  Thus,  the function $ \theta_L( u f_u(A, s))$ is analytic in $S$ for all $s \in [0,t]$ and $A < 0$.
  By the holomorphicity theorem for parameter dependent integrals (see e.g. \cite{KONIG}) we can conclude that
 $
   \int_0^t \theta_L (u f_u(A,s)) ds
$
  is analytic in $S$ for all $A < 0 $.

Defining
$
   \varphi (u, A) := \int_0^t \theta_L (u f_u(A,s)) ds
$
 we now apply \cite{MATT2001} to prove that
$
   \Theta(u) = \int_{\bbr_-} \int_0^t \theta_L (u f_u(A,s)) ds \pi (dA) = \int_{\bbr_-} \varphi(u, A) \pi(dA)
$
  is analytic in $S$. Its conditions $A_1$ and $A_2$ are obviously satisfied. It remains to prove that condition $A_3$ holds, i.e. that $\int_{\bbr_-}  \left| \varphi(u, A) \right| \pi(dA)$ is locally bounded. First observe that
  \begin{eqnarray}\label{bound}
   \left| \theta_L (u f_u(A,s)) \right| 
   & \leq &  \left| \gamma_0 u  f_u(A,s) \right| + \int_{ x \leq 1 } \left| e^{u f_u(A,s) x} -1 \right| \nu(dx) \nonumber \\
    && \quad \quad  + \int_{ x > 1 } \left| e^{u f_u(A,s) x} -1 \right| \nu(dx).
  \end{eqnarray}
  Using (\ref{newbound}), we can bound the first summand in (\ref{bound}) by:
 \[
  | \gamma_0 u f_u(A, s)| 
    \leq |\gamma_0| \left( t \left( |\beta| |u| + \frac{|u|^2}{2} \right) + |\rho| |u| \right) =: B_1(u).
   \]
   For the second summand, using Taylor's theorem we have that
   $
    \left| e^{u f_u(A,s) x} -1 \right| \leq |u f_u(A, s)| |x| + O(|u f_u(A, s)|^2 |x|^2).
   $
   Since
   $
   |u  f_u(A, s)| \leq t \left( |\beta| |u| + \frac{|u|^2}{2} \right) + |\rho| |u|,
   $
   for the remainder term of Taylor's formula we have
   \[
    O(|u f_u(A,s)|^2 |x|^2) \leq O \left( \left| t \left( |\beta| |u| + \frac{|u|^2}{2} \right) + |\rho| |u| \right|^2 |x|^2 \right) ,
   \]
   where the latter term converges to zero as $x \rightarrow 0$.
   If we define
   $$
   K(u) := t \left( |\beta| |u| + \frac{|u|^2}{2} \right) + |\rho| |u|
   $$
   we obtain that
   \begin{flalign*}
    \int_{ x \leq 1 } \left| e^{u f_u(A,s) x} -1 \right| \nu(dx) & \leq K(u) \int_{ x \leq 1 } x \nu(dx) + \int_{ x \leq 1 } O \left( K(u)^2 |x|^2 \right) \nu(dx) =: B_2(u),
   \end{flalign*}
   which is finite due to the properties of the measure $\nu$.

   Let
   $
    S_n := \left\{  \bbc \ni u = v + iw  : \quad |v| \leq \delta -{1}/{n} \right\} \subseteq S.
   $
   Since the function $v f_v(A,s)$ is continuous on the compact set $V_n = \left\{ v \in \bbr: \text{ } |v| \leq \delta - {1}/{n} \right\}$, it attains its minimum and maximum on that set, i.e. there exists $v^* \in V_n$ such that 
$
     v f_v(A,s) \leq v^* f_{v^*} (A,s) \leq |v^* f_{v^*} (A,s)|  =: K_n(u) $ for all $v \in V_n$.
   Note that $v^* \in V_n$ implies that $K_n(u) < \epsilon$. Since $\operatorname{Re}(u f_u(A, s)) \leq v f_v(A,s)$ and $\left| e^{u f_u(A,s) x}\right| = e^{\operatorname{Re}( u f_u(A,s)) x} \leq e^{K_n(u) x}$, it follows that
   \begin{flalign*}
    \int_{ x > 1 } \left| e^{u f_u(A,s) x} -1 \right| \nu(dx) & \leq \int_{ x > 1 } e^{ K_n(u) x} \nu(dx) + \int_{ x > 1 } \nu(dx)  =: B_{3,n}(u),
   \end{flalign*}
   which is finite due to (\ref{condition}) and the properties of the measure $\nu$.

   Since $B_1(u)$, $B_2(u)$ and $B_{3,n}(u)$ do not depend neither on $s$ nor on $A$, we have
$
    |\varphi (u, A)| \leq t (B_1(u) + B_2(u) + B_{3,n}(u))
$
   and
$
    \int_{\bbr_-} t ( B_1(u) + B_2(u) + B_{3, n}(u)) \pi (dA) = t ( B_1(u) + B_2(u) + B_{3,n}(u)) < \infty,
$
    so the function $t (B_1(u) + B_2(u) + B_{3, n}(u))$ is integrable with respect to $\pi$. Since $\varphi (u, A)$ is analytic and thus a continuous function on $S_n$ for all $A < 0 $,  it also holds that $|\varphi (u, A)|$ is continuous on $S_n$ for all $A < 0$. By the dominated convergence theorem it follows that $\int_{\bbr_-}  \left| \varphi(u, A) \right| \pi(dA)$ is continuous and thus a locally bounded function on $S_n$. Since $n \in \bbn$ was arbitrary, it follows that the function is continuous and locally bounded on $S$, which completes the proof.
 \end{proof}
Now we can easily give conditions ensuring that (ii) in Theorem \ref{fourier1} is satisfied. 

\begin{corollary}\label{strip1}
  Let  $\int_{x \geq 1 } e^{rx} \nu(dx) < \infty \quad \text{for all } r \in \bbr \text{ such that } r < \epsilon$
 for some $\epsilon > 0$. Then the moment generating function $\Phi_{X_T|\mathscr{G}_0}$ is analytic on the open strip
$
  S := \{ u \in \bbc : \quad |\operatorname{Re}(u)| < \delta \} 
$
 with
   $
    \delta := -| \beta | - \frac{|\rho|}{T} + \sqrt{\Delta}
 $
   where
 $
    \Delta := \left( |\beta| + \frac{|\rho|}{T} \right)^2 +  \frac{2 \epsilon}{T}.
 $
 Furthermore,
 \begin{align}\label{mgf}
 & \Phi_{X_T|\mathscr{G}_0} (u) = \\
&\,\,\,\exp \left\{ u(X_0 + \mu T) +  \left( u \beta + \frac{1}{2} u^2 \right) \int_{\bbr_-} \int_{-\infty}^0 \frac{1}{A} \left( e^{A(T-s)} - e^{-As} \right) \Lambda(dA, ds) + \Theta (u) \right\}\nonumber
 \end{align}
 for all $u \in S$.
 \end{corollary}
\begin{proof}
 Follows from Theorems \ref{theorem_ch_fn} and  \ref{strip} noting that an analytic function is uniquely identified by its values on a line and \cite[Lemma A.1]{MuhleKarbePfaffelStelzer2009}.
\end{proof}
Very similar to \cite[Th. 6.11]{MuhleKarbePfaffelStelzer2009} we can now prove that also condition (iii) in Theorem~\ref{fourier1} is satisfied for the supOU SV model.
 \begin{theorem}\label{strip3}
 If $u \in \bbc$, $u = v + iw$ and $u \in S$ as defined in Theorem \ref{strip}, then the map
 \begin{equation*}
   w \mapsto \Phi_{X_T|\mathscr{G}_0} (v + iw)
 \end{equation*}
 is absolutely integrable.
\end{theorem}
\section{Examples}\label{sec:4}
\subsection{Concrete specifications}
If we want to price a derivative by Fourier inversion, then this means in the supOU SV model that we have to calculate in general something similar to a three dimensional integral, the inverse Fourier transform and the double integral  in  $\Theta(u) =  \int_{\bbr_-} \int_0^t \theta_L (u f_u(A,s) ) ds \pi(dA)$. If we want to calibrate our model to market data, the optimizer will repeat this procedure very often and so it is important to consider specifications where at least some of the integrals can be calculated analytically.

Actually, it is not hard to see that one can use the standard specifications for $\nu$ of the OU type stochastic volatility model (see  \cite{Schoutens2003,NicolatoVenardos2003,Contetal2004,Barndorffetal2001c}) which are named after the resulting stationary distribution of the OU type processes.

As in the case of  a $\Gamma$-OU process we can choose the underlying L\'{e}vy process to be a compound Poisson process with the characteristic triplet $(\gamma_0, 0, abe^{-bx} \mathbf{1}_{\{ x > 0 \}})$ with $a,b>0$. Furthermore, we assume that $A$ follows a ``negative'' $\Gamma$-distribution, i.e. that $\pi$ is the distribution of $BR$, where $B \in \bbr_-$ and $R \sim \Gamma (\alpha, 1)$ with $\alpha > 1$ which is the specification typically used to obtain long memory/a polynomial decay of the acf.
We refer to this specification as the \emph{$\Gamma$-supOU SV model}.  

Using \eqref{eq:f} we have
\begin{equation*}
 \Theta (u) = u  \int_{\bbr_-} \int_0^t \gamma_0 f_u(A,s) ds \pi(dA) + \int_{\bbr_-} \int_0^t \int_{\bbr_+} \left( e^{u f_u(A,s) x} -1 \right) \nu(dx) ds \pi(dA) .
\end{equation*}

For the first summand in $\Theta(u)$ we see
\begin{multline*}
 u  \int_{\bbr_-} \int_0^t  \gamma _0 f_u(A,s) ds \pi(dA) = \gamma_0 \Bigg( \underbrace{\int_{\bbr_-} \int_0^t \frac{e^{A(t-s)}}{A} \left( u \beta + \frac{u^2}{2} \right) ds \pi(dA)}_{I_1} \\
  - \underbrace{\int_{\bbr_-} \int_0^t \frac{1}{A} \left( u \beta + \frac{u^2}{2} \right) ds \pi(dA)}_{I_2} + \underbrace{\int_{\bbr_-} \int_0^t \rho u ds \pi(dA)}_{I_3} \Bigg).
\end{multline*}
For the three parts we can now show:
\begin{align*}
 I_1 
 & = \left( u \beta + \frac{u^2}{2} \right) \frac{(1 - Bt)^{2 - \alpha} - 1}{B^2(\alpha - 1)(\alpha - 2)}\mbox{ if }\alpha \neq 2,
\\
  I_1 &= - \frac{\left( u \beta + \frac{u^2}{2} \right)}{B^2} \ln(1-Bt)\mbox{ if }\alpha = 2,
\\
 I_2& = \frac{t \left( u \beta + \frac{u^2}{2} \right) }{B (\alpha - 1)},\quad\quad
 I_3 = \rho u \int_0^t \int_{\bbr_-} ds \pi(dA) = \rho u t.
\end{align*}
Furthermore setting $
 C(A) := \frac{1}{A} \left( u \beta + \frac{u^2}{2} \right) - \rho u
$ one obtains for the second summand in $\Theta$
\begin{align*}
&\int_{\bbr_-} \int_0^t \int_{\bbr_+} \left( e^{u f_u(A,s) x} -1 \right) a b e^{-bx} dx ds \pi(dA) \\
&\quad = a \int_{\bbr_-} \frac{1}{A(b+C(A))} \left( b \ln \left( \frac{b - \rho u}{b - \frac{e^{At}}{A} \left( u \beta + \frac{u^2}{2} \right) + C(A)  } \right) - AC(A)t \right) \pi(dA).
\end{align*}
Unfortunately, a more explicit formula for this integral cannot be obtained, and the last integral has to be calculated numerically.

%
 We can also choose the underlying L\'{e}vy process as in an IG-OU model with parameters $\delta$ and $\gamma$, while keeping the choice of the measure $\pi$ the same. In this case we have $\nu(dx)=\frac{1}{2 \sqrt{2 \pi}} \delta \left( x^{-1} + \gamma ^2 \right) x^{-\frac{1}{2}} \exp \left( -\frac{1}{2} \gamma^2 x \right) \mathbf{1}_{ \{x > 0 \} }dx$ and  the only difference compared to the previous case is in the calculation of the triple integral which also can be partially calculated analytically so that only a one-dimensional numerical integration is necessary.
%
%
%

\subsection{Calibration and an illustrative example}

In this chapter, we calibrate the $\Gamma-$supOU SV model to market prices of European plain vanilla call options written on the DAX. 

Let $t_1$, $t_2$,..., $t_M$ be the set of different times to maturity (in increasing order) for which we have market option prices. The parameters to be determined by calibration are $\Pi = (\rho, a, b, B, \alpha, \gamma_0, z_{t_1}, ..., z_{t_M})$, where $\rho$ is the leverage parameter, the parameters $a$ and $b$ are parameters of the measure $\nu$, the parameters $B$ and $\alpha$ are parameters of the measure $\pi$ and $\gamma_0$ is the drift parameter. The parameters $z_{t_1}, ..., z_{t_M}$ are resembling
\begin{equation*}
  z_{t_i} = \int_{\bbr_-} \int_{-\infty}^0 \frac{1}{A} \left( e^{A(t_i-s)} - e^{-As} \right) \Lambda(dA, ds), \quad i = 1,...,M.
\end{equation*}

We calibrate by minimizing the root mean squared error between the Black-Scholes implied volatilities corresponding to market and model prices, i.e. 
\begin{equation*}
  \text{RMSE} = \sqrt{\sum_{i=1}^M \sum_{j=1}^{N_M} \omega_{ij} \left( \text{blsimpv }\big( C_{ij}^M \big) -  \text{blsimpv}\left( C_{ij} \right) \right)^2},
\end{equation*}
where $M$ is the number of different times to maturity, $N_M$ is the number of options for each maturity, $\left\{ C_{ij}^M \right\}$ is the set of market prices and $\left\{ C_{ij} \right\}$ is the set of model prices, $i=1,...,N_M$, $j = 1,...,M$. Of course, minimizing the difference between Black-Scholes implied volatilities is just one possible choice for the objective function. We note that this data example is only supposed to be an illustrative proof of concept and that using other objective functions including in particular weights for the different options should improve the results.

We use closing prices of 200 DAX options on August 19th, 2013. The level of DAX on that day was 8366.29.
The data  source was Bloomberg Finance L.P. and all the options were listed on EUREX. 

For the instantaneous risk-free interest rate we used the 3-month LIBOR rate, which was 0.15173 \%. The maturities of the options were 31, 59, 87, 122, 213, 304, 486 and 668 days. The calibration procedure was performed in MATLAB.

The implied parameters from the calibration procedure are given in Table \ref{dax}. The fit is good: The RMSE  is 0.0046. We plot  market against  model Black-Scholes implied volatilities in Figure \ref{fig:impvol}. Although the RMSE is very low and in plots of market against fitted model prices (not shown here) one sees basically no differences, Figure \ref{fig:impvol} shows that our model fits the implied volatilities for medium and long maturities very well, but the quality of the fit for shorter maturities is lower.  

The vector of the parameters $\{z_{t_i}\}_{i = 1,...,M}$ is indeed increasing with maturity (cf. Lemma \ref{lem:z}), although we actually refrained from including this restriction into our optimization problem. The autocorrelation function of the $\Gamma$-supOU model exhibits long memory for $\alpha \in (1,2)$ (cf. \cite[Section 2.2]{StelzerTosstorff2011}). Since the calibration returns $\alpha = 4.3632$, our market data does not support that long memory is  present. However, $\alpha = 4.3632$ means that the market data is in line with a rather slow polynomial decay of the  autocorrelation function, which is in contrast to the exponential decay of the  autocorrelation function in the OU type SV model. The leverage parameter $\rho$ is negative, which implies a negative correlation between jumps in the volatility and returns. Hence, the typical leverage effect is present. 

The drift parameter of the underlying L\'evy basis $\gamma_0$ is  estimated to be practically zero. So our calibration suggests that a driftless pure jump L\'evy basis may be quite adequate to use.  

If we  compare the supOU stochastic volatility model to the OU type stochastic volatility models (cf. \cite{MuhleKarbePfaffelStelzer2009} or \cite{NicolatoVenardos2003}), we can conclude that neither of the models seems to capture the short-term skew in the implied volatility extremely well. The OU type stochastic volatility models reproduce many of the stylized facts such as jumps in the volatility, (semi-)heavy-tailed distribution of the returns, dependence of the returns without correlation, but they are unable to exhibit long memory (in the squared returns). On the other hand, the supOU stochastic volatility model is able to reproduce additionally 
long memory under certain conditions, but the price that has to be paid for it is that the dimensionality of the optimization problem (the number of parameters) increases with the number of maturities considered.

\begin{table}[t]
\caption{Calibrated parameters for DAX data of August 19th, 2013}\label{dax}
\begin{center}
\begin{tabular}{|c|c|c|c|c|c|}
 \hline
  $\rho$ & $a$ & $b$ & $B$ & $\alpha$ & $\gamma_0$   \\ \hline
  -10.8797 & 0.2225 & 29.4025 & -0.0004 & 4.3632 & 0.0000
  \\ \hline
\end{tabular}
\begin{tabular}{|c|c|c|c|c|c|c|c|}
  \hline
   $z_{t_1}$ & $z_{t_2}$ & $z_{t_3}$ & $z_{t_4}$ & $z_{t_5}$ & $z_{t_6}$ & $z_{t_7}$ & $z_{t_8}$ \\ \hline
   0.0012 & 0.0026 & 0.0038 & 0.0054 & 0.0093 & 0.0136 & 0.0225 & 0.0328
   \\ \hline
\end{tabular}
\end{center}

\end{table}

\begin{figure}[t]
\begin{center}
  \includegraphics[width=0.99\textwidth]{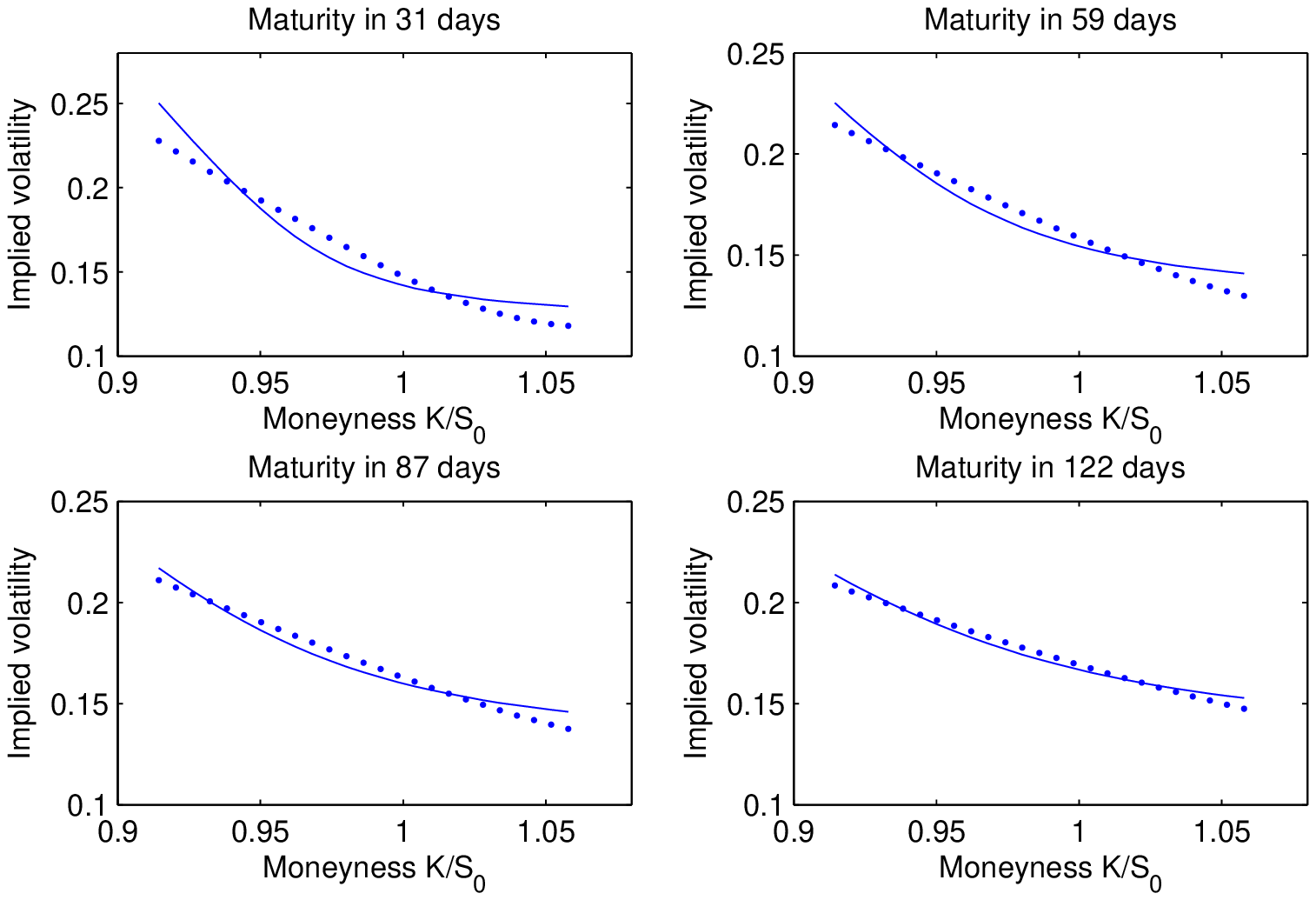}
  \vspace{0cm}
  \includegraphics[width=0.99\textwidth]{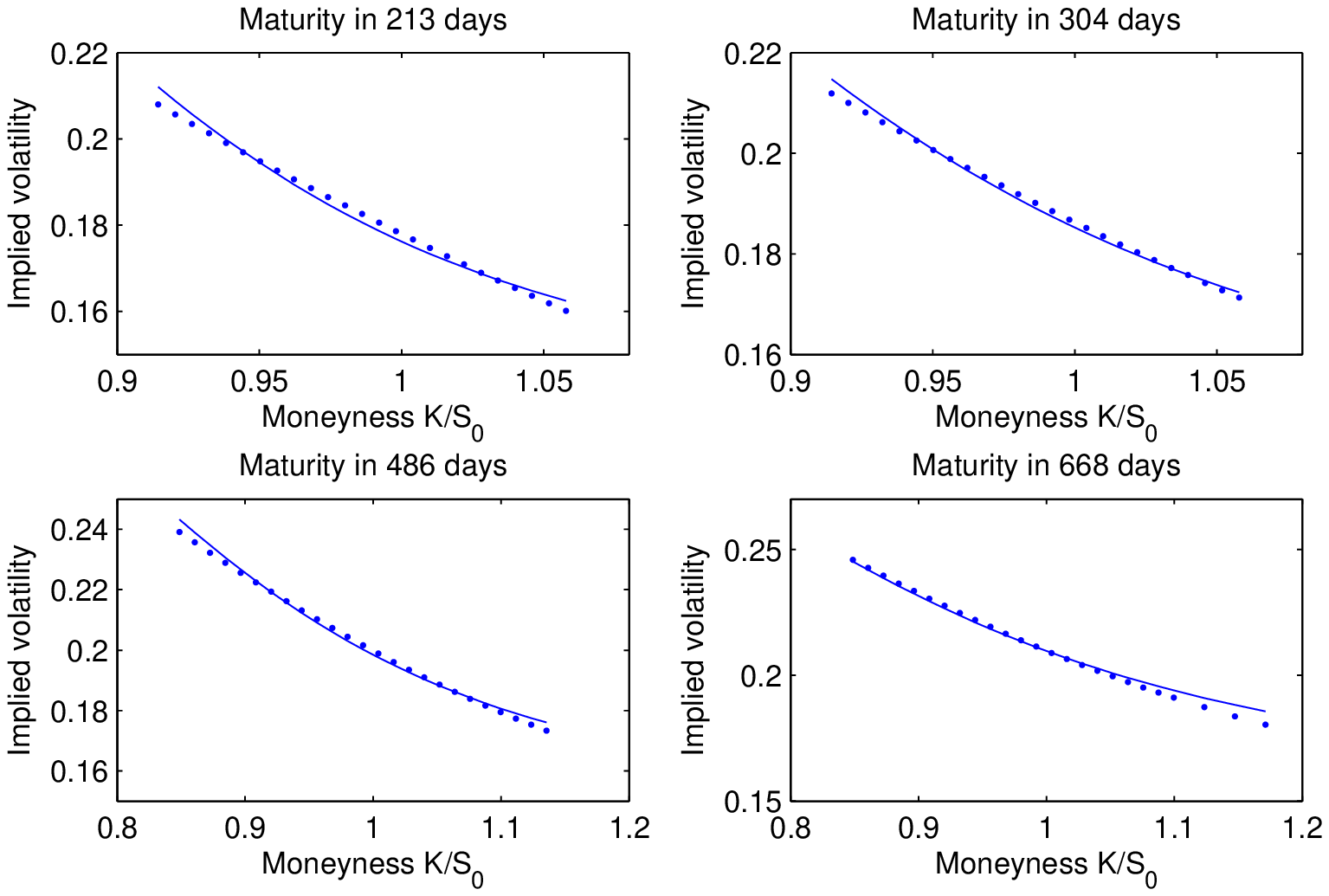}
\end{center}
  \caption{\textbf{Calibration of the supOU model to call options on DAX}: The Black-Scholes implied volatilities. The implied volatilities from market prices are depicted by a dot, the implied volatilities from model prices by a solid line.}\label{fig:impvol}

\end{figure}
\subsection{How to price options with general maturities?}
After having calibrated a model to observed liquid market prices one often wants to use it to price other (exotic) derivatives. Looking at a European derivative with payoff $f(S_T)$ for some measurable function $f$ and maturity $T>0$ one soon realizes that we can only obtain its price directly if $T\in 
\{t_1,t_2,\ldots, t_M\}$ (in other words we can only price derivatives with a maturity for which we have liquid market option prices), as only then we know $z_T$, thus the characteristic function $ \Phi_{X_T|\mathscr{G}_0}$ and therefore the distribution of the price process at time $T$ conditional on our current information $\mathcal{G}_0$. Of course, this is not desirable and the problem is that we do assume that we know $\mathcal{G}_0$ in theory, but in practice we have only limited information in the market prices which we can use to get only parts of the information in $\mathcal{G}_0$.

It seems that to get $z_t$ for  all $t\in\bbr_+$ one needs to really know the whole past of $\Lambda$, i.e. all jumps before time $0$ and the associated times and decay rates. This is clearly not feasible. A detailed analysis on the dependence of $z_t$ on $t$ is beyond the scope of this paper. But we briefly want to comment on possible ad hoc solutions to ``estimate'' $z_T$ based on $\{z_{t_i}\}_{i = 1,...,M}$. The first one is to either interpolate or fit a parametric curve $t\mapsto z_t$ to the ``observed'' $\{z_{t_i}\}_{i = 1,...,M}$. If one also ensures the decreasingness in $t$in this procedure, one should get a reasonably good approximation, especially when the grid given by $\{{t_i}\}_{i = 1,...,M}$ is fine and one considers maturities in $[t_1,t_M]$.

From the probabilistic point of view one would  like to compute $E(z_t|\{z_{t_i}\}_{i = 1,...,M})$ for $t\not \in \{t_1,t_2,\ldots, t_M\}$. Whether and how this conditional expectation can be calculated, is again a question for future investigations. But what one can calculate easily is the best (in the $L^2$ sense) linear predictor of $z_t$ given $\{z_{t_i}\}_{i = 1,...,M}$. One simply needs to straightforwardly adapt standard time series techniques (like the innovations algorithm or linear $L^2$ filtering, see e.g. \cite{Brockwelletal1991}) noting that one has \[\operatorname{cov}(z_t,z_u)=\int_{\bbr_-}\int_{\bbr_-}\frac{e^{-2As}}{A^2}(e^{At}-1)(e^{Au}-1)ds\pi(dA)\int_{\bbr_+}x^2\nu(dx)\,\,\forall \,t,u\in\bbr_+.\]
{\small

}

\end{document}